\documentclass[11pt]{article}
 \usepackage{fullpage}

  \usepackage{amssymb,amsthm,amsmath,hyperref,tikz}
  \usepackage{url}

  \usepackage{latexsym}
  
  \usepackage{enumerate}

\usepackage{graphicx,color} 
\usepackage{chngcntr}

\usepackage{algorithm, algorithmic}

  \newcommand{\C}{\mathbb{C}}

  \newcommand{\R}{\mathbb{R}}

  \newcommand{\U}{\mathbf{U}}
  \newcommand{\uu}{\mathbf{u}}
  
  \newcommand{\w}{\mathbf{w}}
  
  \newcommand{\x}{\mathbf{x}}
  
  \newcommand{\y}{\mathbf{y}}
  
  \newcommand{\z}{\mathbf{z}}
  
  \newcommand{\1}{\mathbf{1}}

  \def\diag{\mathop{{\rm diag}}\nolimits}

  \newcommand{\tr}{\mathop{\mathrm{tr}}\nolimits}
  
  \newcommand{\argmax}{\mathop{\mathrm{arg\,max}}}

  \newcommand{\trans}{^\top}

  \DeclareMathOperator{\prb}{\mathrm{Pr}}
  \DeclareMathOperator{\expval}{\mathrm E}

  \DeclareMathOperator{\poly}{\mathrm{poly}}

  \newtheorem{thm}{Theorem}
  \newtheorem{defn}[thm]{Definition}
  \newtheorem{prop}[thm]{Proposition}
  \newtheorem{prob}[thm]{Problem}
  \newtheorem{lemma}[thm]{ Lemma}

  \usepackage{authblk}

 \numberwithin{equation}{section} 

\begin{document}

 \title{A Simple Spectral Algorithm for Recovering Planted Partitions}
\author{Sam Cole}
\author{Shmuel Friedland}
\author{Lev Reyzin}
\affil{Department of Mathematics, Statistics, and Computer Science\\
 University of Illinois at Chicago \\
Chicago, Illinois 60607-7045, USA \\
\texttt{\{scole3,friedlan,lreyzin\}@uic.edu}}

\begin{titlepage}
\clearpage

 \renewcommand{\thefootnote}{\arabic{footnote}}
\date{}
 \maketitle
\thispagestyle{empty}

\begin{abstract}
In this paper, we consider the \emph{planted partition model}, in which $n = ks$ vertices of a random graph are partitioned into $k$ ``clusters,'' each of size $s$.  Edges between vertices in the same cluster and different clusters are included with constant probability $p$ and $q$, respectively (where $0 \le q < p \le 1$).  We give an efficient algorithm that, with high probability, recovers the clusters as long as the cluster sizes are are least $\Omega(\sqrt{n})$.  Informally, our algorithm constructs the projection operator onto the dominant $k$-dimensional eigenspace of the graph's adjacency matrix and uses it to recover one cluster at a time.  To our knowledge, our algorithm is the first purely spectral algorithm which runs in polynomial time and works even when $s = \Theta(\sqrt n)$, though there have been several non-spectral algorithms which accomplish this.  Our algorithm is also among the simplest of these spectral algorithms, and its proof of correctness illustrates the usefulness of the Cauchy integral formula in this domain.
\end{abstract}
\end{titlepage}

\newpage

\section{Introduction and previous work}\label{intro}
In the Erd\H{o}s-R\'enyi random graph model~\cite{er1959}, graphs $G(n,p)$ on $n$ vertices are generated by including each of the possible $\binom{n}{2}$ edges
independently at random with probability $0 < p < 1$.
A classical conjecture of Karp~\cite{K79} states that there is no efficient algorithm for finding cliques of size~$(1\nobreak+\nobreak\epsilon)\log_{1 / p} n$, though
cliques of size at least $2 \log_{1 / p} n$ will almost surely exist~\cite{bollobas1976cliques}.

Jerrum~\cite{Jerrum92} and Ku\v cera~\cite{K95} introduced a potentially easier variant called the \emph{planted clique} problem.  In this model, one starts with a random graph, but 
additionally, edges are added deterministically to an unknown set of $s$ vertices (known as the ``plant'') to make them form a clique.  
The goal then is to determine a.s. exactly which vertices belong to the planted clique, which should 
be easier when $s$ becomes large.

When $s = \Omega(\sqrt{n \log n})$, the clique can be found by simply taking the $s$ vertices with the largest degrees~\cite{K95}.  This bound
was improved using spectral methods to $\Omega(\sqrt{n})$ by Alon~et~al.~\cite{AKS98} and then others~\cite{AmesV11,Coja-Oghlan10,DekelGP11,FeigeK00,FeigeRon10,McSherry01}.
These methods also handle a generalization of this problem in which edges within the plant are added merely with higher probability rather than deterministically.

A more general version of the problem is to allow for planting multiple disjoint cliques, sometimes called a \emph{planted clustering}.  In the most basic
version, known as the \emph{planted partition} model (also called the stochastic block model), $n$ nodes are partitioned into $k$ disjoint clusters of size $s=n/k$, which 
are ``planted" in a random graph.  Two nodes $u$ and $v$ get an edge with probability $p$ if they are in
the same cluster and with probability $q$ if they reside in different clusters (with $p>q$ constant).  The goal is now to recover the unknown clustering from the random graph generated according to the model, i.e., to determine exactly the vertices in each cluster a.s.

As in the planted clique case, a relatively simple algorithm can recover the clustering when the clique sizes are $\Omega(\sqrt{ n \log n})$---in this
case pairs of vertices with the most common neighbors can be placed in the same cluster~\cite{ChenX14}.  However, when the cluster sizes are only required to be $\Omega(\sqrt{n})$,
the problem, as in the planted clique case, becomes more difficult because a simple application of the Azuma-Hoeffding inequality no longer suffices.  

Our main result is that this can, in fact, be done when the clusters are size $s = \Omega(\sqrt n)$:
\begin{thm}
There exists a deterministic, polytime algorithm which, for sufficiently large $n$, with probability $1-o(1)$ correctly recovers planted partitions in which all clusters are size $s \geq c\sqrt n$, where $c = O(1 / (p - q)^2)$.
\end{thm}
Note that in this paper we consider only the setting in which $p$ and $q$ are constant and all clusters are the same size $s = n / k$.  We discuss more general settings in~\cite{cole2017}. 

Our algorithm is, to our knowledge, the first purely spectral algorithm which runs in polynomial time and recovers the planted partition a.s.\ even when all clusters are size $\Theta(\sqrt n)$, though there have been several non-spectral algorithms which work in this setting~\cite{ames2014,chen2014,oymak2011}.  In particular, the well-known spectral algorithms~\cite{McSherry01, vu2014simple} require that $k = o(\sqrt n)$ and hence do not work when all clusters are size $\Theta(\sqrt n)$ (though they work in considerably more general settings).  On the other hand, Giesen and Mitsche's algorithm~\cite{Giesen05} works when all clusters are size $\Theta(\sqrt n)$ but has running time exponential in $k$.  See Appendix~\ref{prevwork} for comparison with previous work


Efficient algorithms for planted clustering typically rely on either convex optimization \cite{ames2014,chen2014,oymak2011} or spectral techniques \cite{Giesen05,McSherry01,vu2014simple}.  The latter, including ours, 
often involve looking at the projection operator onto the vector space spanned by the $k$ eigenvectors corresponding to the $k$ largest eigenvalues of the adjacency matrix $\hat A$ of the randomly generated graph $\hat G$ and showing that it is ``not too far'' from the projection operator of the expectation matrix $\expval[\hat A]$ onto its own $k$ largest eigenvalues.  Our algorithm is among the simplest of these spectral algorithms: we don't randomly partition the vertices beforehand, and hence there is no messy ``cleanup'' step at the end.  

A natural approach for identifying all the clusters would be to identify a single cluster, remove it, and recurse on the remaining vertices.  This is hard to make work because the randomness of the instance $\hat G$ is ``used up'' in the first iteration, and then subsequent iterations cannot be handled independently of the first.  Existing spectral approaches bypass these difficulties by randomly splitting the input graph into parts, thus forcing
independence in the randomness on the parts~\cite{Giesen05,McSherry01,vu2014simple}.  This partitioning trick works at the cost
of complicating the algorithm.
We, however, are able to make the natural recursive approach work by ``preprocessing the randomness'': we show that certain (exponentially many) events \emph{all} occur simultaneously with high probability, and as long as they all occur our algorithm \emph{definitely} works.

$\Omega(\sqrt n)$ cluster size is generally accepted to be the barrier for efficient algorithms for ``planted''  
problems.  Evidence for the difficulty of beating the $\sqrt{n}$ barrier dates back to Jerrum~\cite{Jerrum92}, who showed
a specific Markov chain approach will fail to find smaller cliques.  Feige and Krauthgamer~\cite{FeigeK00} showed that Lov\'asz-Schrijver SDP
relaxations run into the same barrier, while Feldman~et~al.~\cite{FeldmanGRVX13} show that all ``statistical algorithms" also provably
fail to efficiently find smaller cliques in a distributional version of the planted clique problem.  Recently, Ailon~et~al.\ \cite{ailon2013} were able to recover planted clusterings in which some of the cluster sizes are $o(\sqrt n)$, but their algorithm's success depends on the simultaneous presence of clusters of size $\Omega(\sqrt n\log^2 n)$.

\subsection{Outline}

In Section~\ref{plantedpartition} we formally define the planted partition model.  In Section~\ref{sec:algsection} we present our algorithm for identifying the clusters, and we briefly discuss its running time in Section~\ref{runningtime}.  We prove its correctness in Section~\ref{sec:clusteridproof}.  Sections \ref{randsymmat}-\ref{sec:deviation} are dedicated to developing the linear algebra tools necessary for the proof: in Section~\ref{randsymmat} we introduce tools from random matrix theory which we use in Section~\ref{Ahateigs} to characterize the eigenvalues of the (unknown) expectation matrix $A$ and the randomly generated adjacency matrix $\hat{A}$.  This, in turn, allows us to bound the difference of their projections in Section~\ref{sec:deviation}.  Showing that the projection operators of $A$ and $\hat A$ are ``close'' is the key ingredient in our proof.  


\section{The planted partition problem}\label{plantedpartition}

We now formally define the planted partition problem.

\begin{defn}[Planted partition model]
Let $\mathcal C = \{C_1, \ldots, C_k\}$ be a partition of the set $[n] := \{1, \ldots, n\}$ into $k$ sets of size $s=n/k$, called \emph{clusters} (assume $s | n$).  For constants 
$0 \leq q < p \leq 1$, we define the \emph{planted partition model} $\mathcal G(n, \mathcal C, p, q)$ to be the probability space of graphs with vertex set $[n]$, with edges $ij$ (for $i \neq j$) included independently with probability $p$ if $i$ and $j$ are in the same cluster in $\mathcal C$ and probability $q$ otherwise.  
\end{defn}

See Figure~\ref{plantedpartitionfigure}.  Note that the case $k = 1$ gives the standard Erd\H{o}s-R\'enyi model $\mathcal G(n, p)$~\cite{er1959}, and the case $k = n$ gives $\mathcal G(n, q)$.

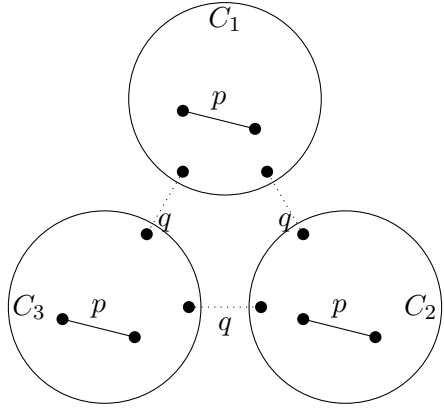
\begin{figure}
\begin{center}
\begin{tikzpicture}[scale=.08]

\foreach \x/\y/\z/\w in {0/34.64/20/0, 20/0/-20/0, -20/0/0/34.64} {	

\draw (\x, \y) circle(16);	

\fill (\x - 7, \y - 2) circle(1);
\draw (\x - 7, \y - 2) -- (\x + 5, \y - 5);
\fill (\x + 5, \y - 5) circle(1);
\draw (\x - 1, \y - 3.5) node[anchor=south]{$p$};

\foreach \t in {.35, .65}
\fill (\t * \x + \z - \t * \z, \t * \y + \w - \t * \w) circle(1);
\draw[dotted] (.35 * \x + .65 * \z, .35 * \y + .65 * \w) -- (.65 * \x + .35 * \z, .65 * \y + .35 * \w);
\draw (.5 * \x + .5 * \z, .5 * \y + .5 * \w) node[anchor=north]{$q$};
}

\draw (0, 51.64) node[anchor=north]{$C_1$};
\draw (37, 0) node[anchor=east]{$C_2$};
\draw (-37, 0) node[anchor=west]{$C_3$};

\end{tikzpicture}
\end{center}
\caption{An illustration of the planted partition model.  Edges between two vertices in the same cluster are added with probability $p$, while edges between two vertices in different clusters are added with probability $q$.}\label{plantedpartitionfigure}

\end{figure}

We will denote as follows the main quantities to consider in this paper. 
\begin{itemize}
\item	$\hat{G} = ([n], \hat E)$ -- a random graph obtained from an \emph{unknown} planted partition distribution $\mathcal G(n, \mathcal C, p, q)$.  This is what the cluster identification
algorithm receives as input. 
\item	$\hat A = (\hat a_{ij})_{i, j = 1}^n \in \{0, 1\}^{n \times n}$ -- the adjacency matrix of $\hat G$.
\item	$\expval[\hat A] := (\expval[\hat a_{ij}])_{i, j = 1}^n$ -- the entrywise expectation of $\hat A$.  \item	$A = (a_{ij})_{i, j= 1}^n := \expval[\hat A] + pI_n$ -- the expectation of the adjacency matrix $\hat{G}$ with $p$s added to the diagonal (to
make it a rank $k$ matrix and simplify the proofs).
 \end{itemize}

\begin{prob}[Planted partition]
Identify (or ``recover'') the unknown partition $C_1, \ldots, C_k$ (up to a permutation of $[k]$) given only $\hat G$, or, equivalently, reproduce $A$ given only $\hat A$.
\end{prob}

In this paper we give an algorithm to recover the clusters which is based on the $k$ largest eigenvalues of  $\hat A$ and the corresponding eigenspaces. 

\subsection{Graph and matrix notation}\label{notation}

We will use the following notation throughout this paper:

\begin{itemize}
\item	$N_G(v)$ -- neighborhood of vertex $v$ in a graph $G$.  We will omit the subscript $G$ when the meaning is clear.
\item	$G[S]$ -- the induced subgraph of $G$ on $S \subseteq V(G)$.
\item	$A[S]$ -- the principal submatrix of $A$ with row and column indices restricted to $S$.
\item	$\lambda_i(A)$ -- the $i$th largest eigenvalue of a symmetric matrix $A$ (recall that symmetric matrices have real eigenvalues).
\item	$\lambda_i(G)$ -- the $i$th largest eigenvalue of $G$'s adjacency matrix.
\item	$P_k(A)$ -- orthogonal projection operator onto the subspace of $\R^n$ spanned by eigenvectors corresponding to the largest $k$ eigenvalues of an $n \times n$ symmetric matrix $A$, represented in the standard basis for $\R^n$.
\item	$||\cdot||_2$ -- the $\ell_2$- (``spectral'') norm of a vector or matrix.
\item	$||\cdot||_F$ -- the Frobenius norm of a matrix.
\item	$I_n$ -- the $n \times n$ identity matrix.
\item	$J_n$ -- the $n \times n$ 1s matrix.
\item	$\1_S$ -- the indicator vector $\in \{0, 1\}^n$ for the set $S \subseteq [n]$.
\item	$\1_n$ -- the all 1s vector $\in \R^n$, i.e. $\1_{[n]}$.
\item	$\expval[X]$ -- the expectation of a random variable $X$.  If $X$ is matrix or vector valued, then the expectation is taken entrywise.
\item	a.s.\ -- almost surely, i.e.\ with probability $1 - o(1)$ as $n \to \infty$.
\end{itemize}

\section{The cluster identification algorithm}\label{sec:algsection}

The main result of this paper is that Algorithm~\ref{alg} below recovers clusters of size $c \sqrt{n}$:

\begin{thm}\label{mainthm}
For sufficiently large $n$ with probability $\geq 1 - 2^{-\Omega(\sqrt n)}$, Algorithm~\ref{alg} correctly recovers planted partitions in which the of size of the clusters is $\geq c\sqrt n$, where $c := \max\left\{\frac{88}{p - q}, \frac{72}{(p-q)^2}\right\}$.
\end{thm}

\begin{algorithm}[H]
\caption{}\label{alg}
Given a graph $\hat G = (\hat V, \hat E)$ and cluster size $s$:
\begin{enumerate}
\item	Let $\hat A$ be the adjacency matrix of $\hat G$, $n := |\hat V|$, $k := n / s$. 

\item	Let $P_k(\hat A) =: (\hat p_{ij})_{i, j \in \hat V}$ be the orthogonal projection operator onto the subspace of $\R^n$ spanned by eigenvectors corresponding to the largest $k$ eigenvalues of $\hat A$.\label{step:proj}

\item	For each column $j$ of $P_k(\hat A)$, let $\hat p_{i_1j} \geq \ldots \geq \hat p_{i_{n - 1}j}$ be the entries other than $\hat p_{jj}$ in nonincreasing order.  Let $W_j := \{j, i_1, \ldots, i_{s - 1}\}$, i.e., the indices of the $s - 1$ greatest entries of column $j$ of $P_k(\hat A)$, along with $j$ itself.\label{step:Wj}

\item	Let $j^*$ be the column $j$ that maximizes $||P_k(\hat A)\1_{W_j}||_2$, i.e. $j^* := \argmax_{j \in \hat V}||P_k(\hat A)\1_{W_j}||_2$.  It will be shown that $W_{j^*}$ has large intersection with a single cluster $C_i \in \mathcal C$ a.s.\label{step:jstar}

\item	Let $C$ be the set of $s$ vertices in $\hat G$ with the \emph{most} neighbors in $W_{j^*}$.  It will be shown that $C = C_i$ a.s.\label{step:id}

\item	Remove $C$ and repeat on $\hat G[\hat V \setminus C]$.  Stop when there are $< s$ vertices left.
\end{enumerate}

\end{algorithm}

The overview of Algorithm~\ref{alg} is as follows.  The algorithm gets a random graph $\hat{G}$ generated according to $\mathcal G(n, \mathcal C, p ,q)$.  
We first construct the projection operator which projects onto the subspace of $\R^n$ spanned by the eigenvectors corresponding to the largest $k$ eigenvalues of $\hat{G}$'s adjacency matrix.  This, we will argue,
gives a fairly good approximation of at least one of the clusters, which we can then find and ``fix up."  Then we remove the 
cluster and repeat the algorithm.

Note that we ensure that Algorithm~\ref{alg} works in every iteration a.s.\ by ``preprocessing the randomness''; more precisely, we will show that a.s.\ certain events occur simultaneously on \emph{all} (exponentially many) subgraphs of $\hat G$ induced on a subset of the clusters, and that as long as they all hold Algorithm~\ref{alg} will \emph{definitely} succeed.  See Section~\ref{sec:clusteridproof}.

\subsection{Running time}\label{runningtime}

Let us analyze the running time of one iteration of Algorithm~\ref{alg}.  Steps~\ref{step:proj} and~\ref{step:jstar} are the most costly.
\begin{itemize}

\item	In step~\ref{step:proj}, computing $P_k(\hat A)$ 
can be done via classical subspace iteration methods in
time $O(n^2k)$~\cite{GolubL96,gu2015subspace}.  Alternatively, one may utilize one of several recent randomized algorithms~\cite{gu2015subspace,HalkoMT11,kishore2016literature,nguyen2009fast} which allow this to be done faster, e.g.\ in time $O(n^2\log k)$~\cite{HalkoMT11}.


\item	
Step~\ref{step:jstar} can be done na\"ively in $O(n^3)$ time. 
However, this can be improved to $O(n^2k)$ by instead multiplying $P_k(\hat A)\hat H$ and taking the norm of each column, where $\hat H$ is defined as in Section~\ref{mainpfnotation}.  From step~\ref{step:proj} we get an orthonormal decomposition of  $P_k(\hat A)$, i.e.\ an $n \times k$ orthogonal matrix $U$ such that  
$UU^\top = P_k(\hat A)$.  Thus, we can compute $P_k(\hat A)\hat H = UU^\top\hat H$ in $O(n^2k)$ time by first multiplying a $k \times n$ matrix and an $n \times n$ matrix, then an $n \times k$ matrix and a $k \times n$ matrix.

In theory, this step can be sped up further using a fast matrix multiplication algorithm~\cite{Coppersmith, LeGall2014}, but such algorithms are rarely used in practice due to numerical instability and large constants hidden in their asymptotic running times. 


\end{itemize}

Thus, each iteration of Algorithm~\ref{alg} can be done in $O(n^2k)$ time.  Since there are $k$ iterations, the overall running time is $O(n^2k^2)$.  In particular, as $k \leq \sqrt n$, this is $O(n^3)$.   


\section{Eigenvalues of random symmetric matrices}\label{randsymmat}

In Section~\ref{Ahateigs} we will show that the eigenvalues of the random matrix $\hat A$ are close to those of its expectation matrix $A$.  To do so, we will need the following well-known result of F\"uredi and Koml\'os about the concentration of eigenvalues of random symmetric matrices~\cite[Theorem~2]{FK81}:

\begin{thm}\label{genFK}
Let $X=[x_{ij}]\in\R^{n\times n}$ be a random symmetric matrix where $x_{ij}$ are independent random variables for $1\le i\le j\le n$.
Assume that there exists $K,\sigma>0$ so that the following conditions hold independent of $n$:
\begin{enumerate}
\item $\expval[x_{ij}]=0$ for $1\le i\le j\le n$.
\item $|x_{ij}|\le K$ for $1\le i\le j\le n$.
\item $\expval[x_{ij}^2] \le \sigma^2$ for $1\le i\le j\le n$.
\end{enumerate}
Then
\begin{equation}\label{genFK1}
\max_{i = 1}^n|\lambda_i(X)| \leq 2\sigma\sqrt{n}+50Kn^{\frac{1}{3}}\log n
\end{equation}
with probability $\geq 1 - n^{-10}$ for $n \geq n_0$.
\end{thm}

Note that the original paper by F\"uredi and Koml\'os assumes that $\expval[x_{ij}^2] = \sigma^2$ for all $i, j$, which in turn makes the bound~(\ref{genFK1}) tight.  However, if all we need is the upper bound in~(\ref{genFK1}), as is the case in this paper, then the proof in~\cite{FK81} goes through with $\expval[x_{ij}^2] \leq \sigma^2$. (Actually, it was pointed out by Vu~\cite{Vu2007} that the proof in~\cite{FK81} contains a minor mistake, so we follow the corrected proof in~\cite{Vu2007}.)

Unfortunately, the $n^{-10}$ failure probability isn't small enough to guarantee our algorithm's success in every iteration, as we will need to apply Theorem~\ref{genFK} simultaneously to $2^{O(\sqrt n)}$ submatrices of $\hat A$ (see Section~\ref{mainpf}); however, we may combine it with the following concentration result to get exponentially small failure probability~\cite[Theorem~1]{alon2002concentration}:

\begin{thm}\label{akv}
Let $X=[x_{ij}]\in\R^{n\times n}$ be a random symmetric matrix where $x_{ij}$ are independent random variables such that $|x_{ij}| \leq 1$ for $1\le i\le j\le n$.  Then for every $1 \leq j \leq n$ the probability that $\lambda_j(X)$ deviates from its median by more than $t$ is at most $4e^{-t^2/32j^2}$.
\end{thm}

Combining Theorems~\ref{genFK} and~\ref{akv}, we get the following:

\begin{thm}\label{strongerFK}
Let $X$ be defined as in Theorem~\ref{genFK}.  Then
\[\max_{i = 1}^n|\lambda_i(X)| \leq 2(\sigma + 3K)\sqrt n\]
with probability $\geq 1 - e^{-n}$ for $n \geq n_0$.
\end{thm}

\begin{proof}
By Theorem~\ref{genFK},
\begin{equation}\label{genFK1m}
\prb\left[\max_{i = 1}^n|\lambda_i(X)| \geq 2\sigma\sqrt{n}+50Kn^{\frac{1}{3}}\log n\right]< \frac{1}{n^{10}}.
\end{equation}
For $n \geq n_0$, we have 
\[50n^{\frac{1}{3}}\log n\le 0.2 \sqrt{n}\Rightarrow 2\sigma\sqrt{n}+50Kn^{\frac{1}{3}}\log n\le 2(\sigma+0.1K)\sqrt{n}.\]
Let $\lambda$ be the median of the random variable $\lambda_1(X)$.  We claim that 
\begin{equation}\label{median}
|\lambda|\le 2(\sigma +0.1 K)\sqrt{n}.
\end{equation}
Indeed,
\[\prb[\lambda_1(X)\ge 2(\sigma +0.1 K)\sqrt{n}] \leq \frac1{n^{10}} \leq \frac12\]  
by~(\ref{genFK1m}).  Now consider the random matrix $-X$.  It satisfies the assumptions of Theorem~\ref{genFK}.  Therefore we have 
\[\prb[\lambda_n(-X)\ge 2(\sigma +0.1 K)\sqrt{n}] \leq \frac1{n^{10}} \leq \frac12.\]
As $\lambda_n(-X) = -\lambda_1(X)$, this is the same as $\prb[\lambda_1(X) \le -2(\sigma +0.1 K)\sqrt{n}]$.  Hence~(\ref{median}) follows by definition of median.

We now ready to apply Theorem~\ref{akv}.  Let $Y=\frac{1}{K}X$.  So now each entry of $Y$ is in $[-1,1]$.  Clearly the median of the random variable $\lambda_1(Y)$ is $\lambda / K$.  By~(\ref{median}) and Theorem~\ref{akv}
\[\prb[\lambda_1(X)\ge 2(\sigma +3K)\sqrt n] \leq \prb\left[\left|\lambda_1(Y)-\frac{\lambda}{K}\right|\ge 5.8\sqrt{n}\right] \leq 4e^{\frac{-(5.8)^2n}{32}}\le \frac12e^{-n}.\]
Similarly, we may apply the entire argument above to $-X$ to get
\[\prb[\lambda_1(-X) \ge 2(\sigma + 3K)\sqrt n] \leq \frac12e^{-n}.\]
Noting that $\max_i|\lambda_i(X)|$ is either $\lambda_1(X)$ or $-\lambda_n(X)=\lambda_1(-X)$, we get 
\[\prb\left[\max_{i = 1}^n|\lambda_i(X)|\ge 2(\sigma +3K)\sqrt{n}\right] \leq e^{-n},\]
as claimed.
\end{proof}

Note that there have been some recent results which give tight bounds on the spectra of more general random matrices~\cite{bandeira2016}, but the above are sufficient for our purposes.

Finally, we will need the following fact from linear algebra:

\begin{prop}[Weyl's inequalities]\label{weyl}
Let $X, Y \in \R^{n \times n}$ be symmetric matrices.  Then
\[\lambda_i(X) + \lambda_n(Y - X) \leq \lambda_i(Y) \leq \lambda_i(X) + \lambda_1(Y - X)\]
for $i = 1, \ldots, n$.  
\end{prop}

See, e.g.,~\cite[Theorem~4.4.6]{friedland2015Matrices} for proof.

\section{Eigenvalues of $A$ and $\hat A$}\label{Ahateigs}
 
The goal of this section is to prove a separation of the first $k$ eigenvalues of both $A$ and $\hat A$ from the remaining $n-k$.   We begin by examining
the eigenvalues of $A$.
 
Without loss of generality, we may assume $C_1 = \{1, \ldots, s\}, C_2 = \{s + 1, \ldots, 2s\}, \ldots, C_k = \{n - s + 1, \ldots, n\}$.  Then the expectation matrix $A$ looks like:
\[A = 
\left(
\begin{array}{ccc|ccc|c|ccc}
p	& \ldots	& p	& q	& \ldots	& q	& 	& q	& \ldots	& q	\\
\vdots	& \ddots	& \vdots	& \vdots	& \ddots		& \vdots	& \cdots	& \vdots	& \ddots	& \vdots	\\
p	& \ldots	& p	& q	& \ldots	& q	& 	& q	& \ldots	& q	\\
\hline
q	& \ldots	& q	& p	& \ldots	& p	& 	& q & \ldots	& q	\\
\vdots	& \ddots	& \vdots	& \vdots	& \ddots		& \vdots	& \cdots	& \vdots	& \ddots	& \vdots	\\
q	& \ldots	& q	& p	& \ldots	& p	& & q & \ldots	& q	\\
\hline
 & \vdots	& 	& 	& \vdots	&	& \ddots	& 	& \vdots	& 	\\
\hline
q	& \ldots	& q	& q	& \ldots	& q	& 	& p	& \ldots	& p	\\
\vdots	& \ddots	& \vdots	& \vdots	& \ddots	& \vdots	& \cdots	& \vdots	& \ddots	& \vdots	\\
q	& \ldots	& q	& q	& \ldots	& q	& 	& p	& \ldots	& p	
\end{array}
\right) = qJ_n + (p - q)\diag(J_s, \ldots, J_s),\]
where $J_m$ is the $m \times m$ ones matrix.  Thus, $A$ contains all the information about the unknown partition $\mathcal C$.

The following lemma is easily verified: 
 
\begin{lemma}\label{lemma:Aeigs}
The eigenvalues of $A$ are
\begin{eqnarray*}
&&\lambda_1(A) = (p-q)s + qn,\\
&&\lambda_i(A) = (p-q)s\ \mathrm{for}\ i = 2,\ldots,k,	\\
&&\lambda_i(A) = 0\ \mathrm{for}\ i = k + 1, \ldots, n.
\end{eqnarray*}
\end{lemma}

So we see that the smallest positive eigenvalue is proportional to the size of the clusters.

We continue by bounding the spectral norm of $\hat A - A$ (recall that the spectral norm of a symmetric matrix $X \in \R^{n \times n}$ is $||X||_2 = \max_{i = 1}^n|\lambda_i(X)|$; see~\cite[Corollary~4.11.13]{friedland2015Matrices}).
\begin{lemma}\label{normAhat-A}
For sufficiently large $n$, 
\begin{equation}\label{l2normbound}
\|\hat A - A\|_2 \leq 8 \sqrt{n}
\end{equation}
with probability $\geq 1 - e^{-n}$.
\end{lemma}

\begin{proof}
Set $X=(x_{ij})=\hat{A}-\expval[\hat A]$.  Let $\sigma_{ij}$ be the standard deviation of $x_{ij}$ and let $\sigma \ge \sigma_{ij} \mathrm{\ for\ } i,j \in [n]$.  Hence, $X$ satisfies the conditions of Theorem~\ref{strongerFK}, with 
\begin{equation*}\label{vaKsigma}
K=1, \quad \sigma=\max(\sqrt{p(1-p)}, \sqrt{q(1-q)})\le \frac{1}{2}.
\end{equation*}
Thus, 
\begin{equation}\label{applyFK}
||X||_2 = \max_{i = 1}^n|\lambda_i(X)| \leq 2(\sigma + 3K)\sqrt n \leq 7\sqrt n
\end{equation}
with probability $\geq 1 - e^{-n}$ by Theorem~\ref{strongerFK}.

Observe that
\[
\hat{A}-A=\hat A-\expval[\hat A]- p I_{n}=X- pI_{n}
\]
\[
\Rightarrow
\]
\[
\|\hat A-A\|_2=\|X - pI_n\|_2 = \|X\|_2  + \|pI_n \|_2 \le \|X\|_2 +p.
\]
From~(\ref{applyFK}) we deduce that 
\[||\hat A - A||_2 \leq 7\sqrt n + p \leq 8\sqrt n\]
with probability $> 1 - e^{-n}$ for $n>n_0$.  
\end{proof}

We can now use the lemma above to characterize the eigenvalues of $\hat{A}$ (and $A$) as follows:

\begin{lemma}\label{Ahat-eigs}
Assume $\hat A$ satisfies~(\ref{l2normbound}) and $s \geq c\sqrt n$.  Then the largest $k$ eigenvalues of $A$ and $\hat{A}$ are in the interval $[c'\sqrt n,n]$ and all other eigenvalues of $A$ and $\hat{A}$ are in the interval $[-8\sqrt{n},8\sqrt{n}]$, where
\begin{equation}\label{cprimedef}
c' := (p - q)c - 8.
\end{equation}
\end{lemma}

\begin{proof}
Applying Weyl's inequalities (Proposition~\ref{weyl}) to Lemma~\ref{normAhat-A} yields
\[|\lambda_i(\hat{A})-\lambda_i(A)|\le \max(|\lambda_n(A - \hat A)|, |\lambda_1(A - \hat A)|) = ||A - \hat A||_2 \leq 8\sqrt n\]
for $i = 1, \ldots, n$.  Thus, by Lemma~\ref{lemma:Aeigs} we get
\begin{eqnarray*}
\lambda_i(\hat A)	& \geq	& (p - q)s - 8\sqrt n 	\\
					& \geq 	& ((p - q)c - 8)\sqrt n \textrm{ for } i = 1, \ldots, k,\\
|\lambda_i(\hat A)|	& \leq	& 8\sqrt n \textrm{ for } i = k + 1, \ldots, n.
\end{eqnarray*}
The lemma  thus follows by definition of $c'$.

Note that the upper bound of $n$ follows from the fact that for any $X = (x_{ij}) \in \R^{n \times n}$ we have $\lambda_1(X) \leq \max_i\sum_j|x_{ij}|$.
\end{proof}

Lemma~\ref{Ahat-eigs} shows that a.s.\ we have a separation between the largest $k$ eigenvalues and the remaining eigenvalues of both $A$ and $\hat A$, provided that $c' > 8$, or equivalently
\begin{equation}\label{creq1}
c > \frac{16}{p - q}.
\end{equation}
We will assume this is the case from now on.  

Note that the argument above shows that, in fact, $\lambda_1(\hat A) \geq qn + (p - q)s - 8\sqrt n = \Theta(n)$, while $\lambda_2(\hat A), \ldots, \lambda_k(\hat A) \leq (p - q)s + 8\sqrt n = O(s)$, but this information will not be needed hence.  Figure~\ref{eigsfigure} illustrates the distribution of eigenvalues of $A$ and $\hat A$.

\begin{figure}\centering
\begin{tikzpicture}[scale=.08]

\draw[<->] (-25, 0) -- (125, 0);

\foreach \x/\lbl in {-20/{$-8\sqrt n$}, 0/{0}, 20/{$8\sqrt n$}, 30/{}, 120/{$n$}}
\draw (\x, 2) -- (\x, -2) node[anchor=north] {\lbl};
\draw (35, -2) node[anchor=north] {$c'\sqrt n$};

\foreach \x in {0, 35, 110}
\draw (\x, 0) circle(1);
\foreach \x in {-4.8, 10, 16.8, -14.2, -7.2, 7.4, 17.8, 31, 45, 47, 49, 53, 107}
\fill (\x, 0) circle(1);

\draw (108.5, 2) node[anchor=south]{$\lambda_1$};
\draw (42, 2) node[anchor=south]{$\lambda_2, \ldots, \lambda_k$};
\draw (0, 2) node[anchor=south]{$\lambda_{k + 1}, \ldots, \lambda_n$};

\end{tikzpicture}
\caption{The distribution of eigenvalues of $A$ ($\circ$) and $\hat A$ ($\bullet$).}\label{eigsfigure}

\end{figure}
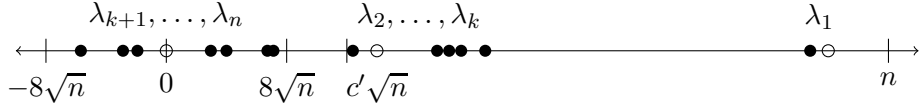

\section{Deviations between the projectors $P_k(\hat{A})$ and $P_k(A)$}\label{sec:deviation}

In this section, we will prove bounds on $\|P_k(\hat{A})-P_k(A)\|_2$ and $\|P_k(\hat{A})-P_k(A)\|_F$, where $\|\cdot\|_2$ and $\|\cdot\|_F$ are
the spectral and the Frobenius matrix norms, respectively.  The following lemma characterizes $P_k(A)$:

\begin{lemma}
\begin{equation}\label{Hdef}
P_k(A) = \frac1s\sum_{i = 1}^k\1_{C_i}\1_{C_i}\trans = \frac1sH,
\end{equation}
where $H \in \{0, 1\}^n$ is the ``true'' cluster matrix whose $(i, j)$th entry is 1 if and only if $i$ and $j$ are in the same cluster.
\end{lemma}

\begin{proof}
Let $\uu_i := \frac1{\sqrt s}\1_{C_i} \in \R^n$ for $i = 1, \ldots, k$, and let $\U$ be the subspace of $\R^n$ spanned by eigenvectors corresponding to $\lambda_1(A), \ldots, \lambda_k(A)$.  It is easily verified that $\uu_1, \ldots, \uu_k$ are an orthonormal basis for $\U$.  Thus, letting $P_\U$ denote the orthogonal projection operator onto $\U$, we get
\[P_k(A) = P_\U = \sum_{i = 1}^k\uu_i\uu_i\trans = \frac1s\sum_{i = 1}^k\1_{C_i}\1_{C_i}\trans.\qedhere\]
\end{proof}

If we assume $C_1 = \{1, \ldots, s\}, C_2 = \{s + 1, \ldots, 2s\}, \ldots, C_k = \{n - s + 1, \ldots, n\}$ as in Section~\ref{Ahateigs}, then $P_k(A)$ looks like:
\[P_k(A) = \frac1s\left(
\begin{array}{ccc|ccc|c|ccc}
1	& \ldots	& 1	& 0	& \ldots	& 0	& 	& 0	& \ldots	& 0	\\
\vdots	& \ddots	& \vdots	& \vdots	& \ddots		& \vdots	& \cdots	& \vdots	& \ddots	& \vdots	\\
1	& \ldots	& 1	& 0	& \ldots	& 0	& 	& 0	& \ldots	& 0	\\
\hline
0	& \ldots	& 0	& 1	& \ldots	& 1	& 	& 0 & \ldots	& 0	\\
\vdots	& \ddots	& \vdots	& \vdots	& \ddots		& \vdots	& \cdots	& \vdots	& \ddots	& \vdots	\\
0	& \ldots	& 0	& 1	& \ldots	& 1	& & 0 & \ldots	& 0	\\
\hline
 & \vdots	& 	& 	& \vdots	&	& \ddots	& 	& \vdots	& 	\\
\hline
0	& \ldots	& 0	& 0	& \ldots	& 0	& 	& 1	& \ldots	& 1	\\
\vdots	& \ddots	& \vdots	& \vdots	& \ddots	& \vdots	& \cdots	& \vdots	& \ddots	& \vdots	\\
0	& \ldots	& 0	& 0	& \ldots	& 0	& 	& 1	& \ldots	& 1	
\end{array}
\right) = \frac1s\diag(J_s, \ldots, J_s)
\]
when represented in the standard basis for $\R^n$.

So we see that the columns of $P_k(A)$ are essentially the indicator vectors of the unknown clusters $C_1, \ldots, C_k$.  The central idea behind Algorithm~\ref{alg} is that if $||P_k(A) - P_k(\hat A)||_F$ is sufficiently small, then some column of $P_k(\hat A)$ is a good approximation to the corresponding column of $P_k(A)$ and can thus be used to recover the corresponding cluster.

\subsection{The Cauchy integral formula for projections}

To prove such bounds on $||P_k(A) - P_k(\hat A)||_2$ and $||P_k(A) - P_k(\hat A)||_F$ we will employ the Cauchy integral formula.  Similar applications
of the Cauchy integral formula are studied in matrix perturbation theory~\cite{GolubL96,Stewart73} and could be adapted to obtain our bound on $||P_k(A) - P_k(\hat A)||_2$, but we include the full proof for the sake of exposition.

Recall that an analytic function $f : \C \to \C$ can be extended to a function of matrices via its Taylor series~\cite{friedland2015Matrices}:
\[f(Z) := f(a)I_n + \frac{f'(a)}{1!}(Z - aI_n) + \frac{f''(a)}{2!}(Z - aI_n)^2 + \ldots.\]
In particular, if $Z$ is diagonalizable as $Z = PDP^{-1}$ (as is any symmetric matrix), then $f(Z) = Pf(D)P^{-1}$, where $f(D)$ is evaluated by simply applying $f$ to each diagonal entry.

Accordingly, we also get an extension of the Cauchy integral formula to matrices~\cite[Theorem~3.4.2]{friedland2015Matrices}:

\begin{thm}
Let $\Omega$ be an open set in $\C$.  Assume that $\Gamma$ is a finite set of disjoint simple, closed curves such that $\Gamma$ is the boundary of an open set $D$, and $\Gamma \cup D \subset \Omega$.  Assume that $Z \in \C^{n \times n}$ and $\lambda_i(Z) \in D$ for $i = 1, \ldots, n$.  Then for any $\phi : \C \to \C$ analytic on $\Omega$
\[\phi(Z) = \frac1{2\pi i}\int_\Gamma(zI_n - Z)^{-1}\phi(z)dz.\]
\end{thm}

We get the following as a corollary~\cite[Problem~3.4.10]{friedland2015Matrices}:

\begin{thm}\label{cauchythm}
Let $B \in \R^{n \times n}$ be a symmetric matrix.  Let $\gamma \subset \C$ be a simple, closed curve which is the boundary of an open set $D$ such that $\lambda_1(B), \ldots, \lambda_k(B) \in D$ and $\lambda_{k + 1}(B), \ldots, \lambda_n(B) \notin D \cup \gamma$.  Then
\[P_k(B) = \frac1{2\pi i}\int_\gamma(zI_n - B)^{-1}dz.\]
\end{thm}


\subsection{A bound on $\|P_k(\hat A)-P_k(A)\|_2$}

As $P_k(\hat{A})$ and $P_k(A)$ are projection operators, we have
\[|P_k(\hat{A})\|_2=\|P_k(A)\|_2=1 \quad \Rightarrow \quad \|P_k(\hat{A})-P_k(A)\|_2\le 2.\]
In fact, we can make this difference arbitrarily small by increasing the cluster size appropriately, as shown in the following lemma.

\begin{lemma}\label{fundestdifprojnorm}  
Assume $\hat A$ satisfies~(\ref{l2normbound}) and $s \geq c\sqrt n$.  Then
\begin{equation}\label{defPkABeps}
\|P_k(\hat{A})-P_k(A)\|_2\le \epsilon
\end{equation}
if $c$ is sufficiently large.
\end{lemma}

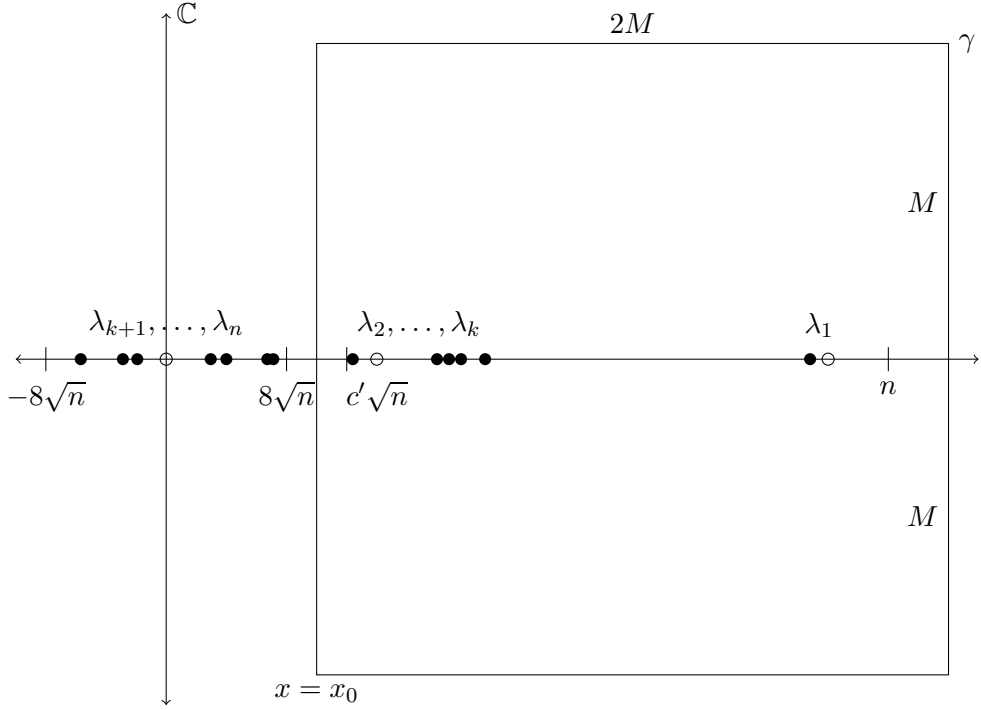
\begin{figure}
\centering
\begin{tikzpicture}[scale=.08]

\draw[<->] (-25, 0) -- (135, 0);
\draw[<->] (0, -57.5) -- (0, 57.5) node[anchor=west]{$\C$};

\foreach \x/\lbl in {-20/{$-8\sqrt n$}, 20/{$8\sqrt n$}, 30/{}, 120/{$n$}}
\draw (\x, 2) -- (\x, -2) node[anchor=north] {\lbl};
\draw (35, -2) node[anchor=north] {$c'\sqrt n$};

\draw (25, -52.5) rectangle (130, 52.5) node[anchor=west]{$\gamma$};
\draw (77.5, 52.5) node[anchor=south]{$2M$};
\foreach \y in {26.125, -26.125}
\draw (130, \y) node[anchor=east]{$M$};
\draw (25, -52.5) node[anchor=north]{$x = x_0$};

\foreach \x in {0, 35, 110}
\draw (\x, 0) circle(1);
\foreach \x in {-4.8, 10, 16.8, -14.2, -7.2, 7.4, 17.8, 31, 45, 47, 49, 53, 107}
\fill (\x, 0) circle(1);

\draw (108.5, 2) node[anchor=south]{$\lambda_1$};
\draw (42, 2) node[anchor=south]{$\lambda_2, \ldots, \lambda_k$};
\draw (0, 2) node[anchor=south]{$\lambda_{k + 1}, \ldots, \lambda_n$};
\end{tikzpicture}

\caption{The largest $k$ eigenvalues of both $A$ ($\circ$) and $\hat A$ ($\bullet$) are in the interior of $\gamma$, while the remaining $n - k$ eigenvalues are in the exterior.}\label{gamma}
\end{figure}

\begin{proof}
Define $\gamma$ to be a square in the complex plane with the length $2M\gg m$.  Its  sides are parallel to the $x$- and $y$-axes. 
The center of of the square is on the $x$-axis.
The left and right sides of the square are on the lines $x=x_0 := \frac{(c' + 8)\sqrt n}{2}$ and $x=x_0+2M$, respectively, where $c'$ is defined in~(\ref{cprimedef}).  The upper and lower sides of the square are on the lines $y = \pm M$.  Note that by Lemma~\ref{Ahat-eigs} the interior of $\gamma$ contains the $k$ largest eigenvalues of $A$ and $\hat A$ and the exterior of $\gamma$ contains the other $n - k$ eigenvalues of $A$ and $\hat A$ (see Figure~\ref{gamma}).  To get our estimate \eqref{defPkABeps} we will let $M\to \infty$.

Applying Theorem~\ref{cauchythm},
\begin{eqnarray*}\label{Caucintfor}
P_k(\hat A)=\frac{1}{2\pi{i}}\int_{\gamma} (zI_n-\hat A)^{-1}dz,\\ 
P_k(A)=\frac{1}{2\pi{i}}\int_{\gamma} (zI_n-A)^{-1}dz.\notag
\end{eqnarray*}
Hence
\begin{eqnarray*}
P_k(\hat A) - P_k(A)&=&\frac{1}{2\pi{i}}\int_{\gamma} (zI_n-\hat A)^{-1}\big((zI_n-A)-(zI_n-\hat A)\big)(zI_n-A)^{-1} dz\\
&=&\frac{1}{2\pi{i}}\int_{\gamma} (zI_n-\hat A)^{-1}\big(\hat A-A\big)(zI_n-A)^{-1} dz,
\end{eqnarray*}
and so we get
\begin{eqnarray}
\|P_k(\hat A) - P_k(A)\|_2&\le& \frac{1}{2\pi}\int_{\gamma} \|(zI_n-\hat A)^{-1}\big(\hat A-A\big)(zI_n-A)^{-1}\| _2|dz| \label{Cauchyforest} \\
&\le& \frac{1}{2\pi}\int_{\gamma} \|(zI_n-\hat A)^{-1}\|_2\|\hat A-A\|_2 \|(zI_n-A)^{-1}\|_2 |dz|.\notag
\end{eqnarray}
Observe that for each $z\in\C$ the matrices $zI_n-\hat A, zI_n-A $ are normal.
Hence 
\[\|(zI_n-\hat{A})^{-1}\|_2=\frac{1}{\min_{j\in [n]} |z-\lambda_j(\hat{A})|}, \quad  \|(zI_n-A)^{-1}\|_2=\frac{1}{\min_{j\in [n]} |z-\lambda_j(A)|}.\]

  Let us first estimate the contribution to the integral \eqref{Cauchyforest} on the left side of $\gamma$.  Let $z=x_0 +y{i}, y\in\R$.  That is, $z$ lies on the line $x=x_0$.
Therefore, by Lemma~\ref{Ahat-eigs}
\[|z-\lambda_j(A)|, |z - \lambda_j(\hat A)| \ge \sqrt {\left(\frac{(c' - 8)\sqrt n}2\right)^2 +y^2} \textrm{ for } z=x_0 +y{i}.\]
Also recall from~(\ref{l2normbound}) that $\|\hat{A}-A\|_2\le 8\sqrt{n}$. Hence for  $z=x_0 +y{i}$ one has the estimate:
\[
\frac{1}{2\pi}\int_{-M}^{M} \Vert(zI_n-\hat{A})^{-1}\Vert_2  \\
\Vert (\hat{A}-A)\Vert_2 \Vert( zI_n-A)^{-1}\Vert_2 |dz|
\]
\[
<
\]
\begin{eqnarray}\notag
\frac{1}{2\pi}\int_{-\infty}^{\infty}\frac{8\sqrt{n}}{\left(\frac{(c' - 8)\sqrt n}2\right)^2 +y^2} dy =\frac{8}{\pi(c'-8)}\int_{-\infty}^\infty
\frac{du}{1+u^2}
=\frac{8}{(c'-8)}.\label{est1int}
\end{eqnarray}

Next we estimate the contribution of the integral \eqref{Cauchyforest} on the other three sides.   Consider first the side on the line $y=M$.
Since the eigenvalues of $\hat{A}$ and $A$ are real it follows that 
\[
|(x+M{i})I_n-\hat{A}\|_2, \|(x+M{i})I_n-A\|_2\ge M
\]
\[\Rightarrow\]
\[
|((x+M{i})-\hat{A})^{-1}\|_2, \|((x+M{i})-A)^{-1}\|_2\le\frac{1}{M}.
\]
Hence the contribution of \eqref{Cauchyforest} on the upper side of the square is bounded above by $\frac{8M\sqrt{n}}{2\pi M^2}=\frac{4\sqrt{n}}{\pi M}$.
The same upper estimate holds for the lower side of the square on the line $y=-M$.  We now estimate from above the contribution of \eqref{Cauchyforest}
on the right side of the square on $x=x_0+2M$.  Since the eigenvalues of $\hat{A}$ and $A$ are real and at most $n$ it follows that
\[
||(x_0+2M +y{i})I_n-\hat{A}||_2, ||(x_0+2M +y{i})I_n-A||_2 \ \ge \ 2M-n
\]
\[
\Rightarrow
\]
\[
||((x_0+2M +y{i})I_n-\hat{A})^{-1}||_2, ||((x_0+2M +y{i})I_n-A)^{-1}||_2  \ \le  \ \frac{1}{2M-n}.
\]
Hence the contribution of \eqref{Cauchyforest} on the righthand side of the square is bounded above by $\frac{4M\sqrt{n}}{\pi(2M-n)^2}$.
Therefore
\[\|P_k(\hat A)-P_k(A)\|_2\le \frac{8}{(c'-8)}+2 \cdot \frac{4\sqrt{n}}{\pi M}+\frac{4M\sqrt{n}}{\pi(2M-n)^2}.\]
Letting $M\to\infty$ we obtain $||P_k(\hat A) - P_k(A)||_2 \leq \frac8{c' - 8}$.
We therefore get~\eqref{defPkABeps} if 
\begin{equation}\label{sreq2}
\frac8{c' - 8} = \frac8{(p - q)c - 16} \leq \epsilon,
\end{equation}
completing the proof.
\end{proof}

\subsection{A bound on $\|P_k(\hat{A})-P_k(A)\|_F$}
Now we estimate the Frobenius norm of $P_k(\hat{A})-P_k(A)$.  Recall that for any matrix $B=(b_{ij})\in\R^{m\times n}$ 
\[\|B\|_F := \sqrt{\sum_{i=j=1}^{m,n} b_{ij}^2}.\]
Moreover, if $B$ is a symmetric, then
\begin{equation}\label{frobsymmetric}
\|B\|_F^2=\sum_{i=1}^n \lambda_i(B)^2.
\end{equation}
Therefore we obtain the following lemma:
\begin{lemma}\label{estfrobnormGhat}
$||P_k(\hat{A})-P_k(A)||_F^2 \leq 2k||P_k(\hat{A})-P_k(A)||_2^2$.
\end{lemma}

\begin{proof}
Recall that $P_k(\hat{A})$ and $P_k(A)$ have rank $k$.
Hence $P_k(\hat{A})-P_k(A)$ has rank at most $2k$.  So $P_k(\hat{A})-P_k(A)$ has at most $2k$ nonzero eigenvalues.  The lemma thus follows from~(\ref{frobsymmetric}).
\end{proof}

\section{Proof of algorithm's corectness}\label{sec:clusteridproof}

The proof of Algorithm~\ref{alg}'s correctness goes roughly as follows.  We will prove using the spectral analysis in Sections~\ref{Ahateigs}-\ref{sec:deviation} that a.s.\ there is a column $j$ for which $||P_k(\hat A)\1_{W_j}||_2$ is ``large'' (Lemma~\ref{existsgoodcol}).  Next, we will show that any for any such $j$, the set $W_j$ consists mostly of vertices from a single cluster (Lemma~\ref{almostcluster}).  Finally, we show how to recover this cluster exactly by looking at how many neighbors each vertex has in $W_j$ (Lemmas~\ref{recoverclusterswhp}-\ref{recovercluster}).

This argument shows that Algorithm~\ref{alg} succeeds in iteration 1 a.s.  To show that it succeeds in \emph{every} iteration, we will apply the same argument to all ``cluster subgraphs'' of $\hat G$---i.e., those subgraphs induced on a subset of the clusters.  We will prove that all $2^k$ such subgraphs have certain desirable properties a.s., in which case our algorithm \emph{deterministically} succeeds in identifying a cluster.  Therefore, when we remove it we are considering another cluster subgraph, so the algorithm again succeeds, and so on.  Thus, we are able to restrict our analysis to these cluster subgraphs, bounding the number of events that need to occur in order to ensure the algorithm's success.  This is how we avoid the need to randomly split the graph into parts,  as in~\cite{Giesen05,McSherry01,vu2014simple}.  The details of this approach, which we call ``preprocessing the randomness,'' are presented in Section~\ref{mainpf}

\subsection{Notation}\label{mainpfnotation}

We will use the following notation in our proof:
\begin{itemize}
\item	$H = (h_{ij})_{i, j = 1}^n$ -- the ``true cluster matrix'' as defined in~(\ref{Hdef}), i.e., $h_{ij} = 1$ if $i$ and $j$ are in the same cluster, 0 else.

\item	$W_1, \ldots, W_n$ as defined as in step~\ref{step:Wj} of Algorithm~\ref{alg}.

\item	$\hat H = (\hat h_{ij})_{i, j = 1}^n := (\1_{W_1}, \ldots, \1_{W_n})$ -- the ``estimated cluster matrix.''  The idea is that at least one column of $\hat H$ will be a good approximation of the corresponding column of $H$, and we will give a way to find such a column.  Note that each column of $\hat H$ has exactly $s$ 1s and that $\hat H$ need not be symmetric.
\end{itemize}

\subsection{Technical lemmas}\label{techlemmas}

The proof of Theorem~\ref{mainthm} relies on several additional lemmas.  Lemmas~\ref{existsgoodcol}-\ref{recovercluster} fit together roughly as follows:
\begin{itemize}
\item	Lemma~\ref{existsgoodcol} says that a.s.\ there is a column $j$ for which $||P_k(\hat A)\1_{W_j}||_2$ is large.
\item	Lemma~\ref{almostcluster} says that for such a column $j$, $W_j$ consists mostly of vertices from a single cluster $C_i$.
\item	Lemmas~\ref{recoverclusterswhp} and \ref{recovercluster} say that a.s.\ vertices in $C_i$ will have many neighbors in $W_j$, while vertices outside $C_i$ will have relatively few neighbors in $W_j$; hence, we can recover $C_i$ by taking the $s$ vertices with the most neighbors in $W_j$.
\end{itemize}

\begin{lemma}\label{existsgoodcol}
Assume $\hat A$ satisfies~(\ref{l2normbound}).  Then there exists a column $j$ such that 
\begin{equation}\label{goodcoldef}
||P_k(\hat A)\1_{W_j}||_2 \geq (1 - 8\epsilon^2 - \epsilon)\sqrt s.
\end{equation}
\end{lemma}

\begin{proof}
Lemmas~\ref{fundestdifprojnorm} and \ref{estfrobnormGhat} together give
\begin{equation*}
||P_k(A) - P_k(\hat A)||_F^2 \leq 2k\epsilon^2.
\end{equation*}
By definition of $\hat H$,
\[\tr(H^2) = \tr(\hat H\trans\hat H) = ns\]
and, letting $P_k(\hat A) = (\hat p_{ij})_{i, j = 1}^n$, for each column $j \in [n]$ we have
\[(\hat H\trans P_k(\hat A))_{jj} = \sum_{i = 1}^n\hat h_{ij}\hat p_{ij} \geq \sum_{i = 1}^nh_{ij}\hat p_{ij} = (HP_k(\hat A))_{jj} \quad \Rightarrow \quad \tr(\hat H\trans P_k(\hat A)) \geq \tr(HP_k(\hat A)).\]
Recall also~(\ref{Hdef}) that $P_k(A) = \frac1sH$.  Therefore
\begin{eqnarray*}
2k\epsilon^2	& \geq	& ||P_k(\hat A) - P_k(A)||_F^2	\\
& = 		& \|\frac1sH - P_k(\hat A)\|_F^2		\\
& =		& \frac{1}{s^2} \tr(H^2)  +\tr(P_k(\hat A)^2) - 2\frac{1}{s}\tr (HP_k(\hat A))	\\
& \geq	& \frac{1}{s^2}  \tr(\hat H\trans \hat H)  +\tr(P_k(\hat A)^2) -2\frac{1}{s}\tr (\hat H\trans P_k(\hat A))	\\
& =		& \|\frac1s\hat H - P_k(\hat A)\|_F^2.
\end{eqnarray*}
The triangle inequality then yields:
\[\|\frac{1}{s}H-\frac{1}{s}\hat H\|_F\le \|\frac{1}{s}H-P_k(\hat A)\|_F +\|\frac{1}{s}\hat H-P_k(\hat A)\|_F\le 2\epsilon\sqrt{2k}.
\]
Thus 
\begin{equation*}
||H - \hat H||_F^2 = \sum_{j = 1}^n\left(\sum_{i = 1}^n(h_{ij} - \hat h_{ij})^2\right) \leq 8k\epsilon^2s^2,
\end{equation*}
so by averaging there exists a column $j^*$ such that
\begin{equation}\label{jstar}
\sum_{i = 1}^n(h_{ij^*} - \hat h_{ij^*})^2 \leq \frac1n \cdot 8k\epsilon^2s^2 = 8\epsilon^2s.
\end{equation}

Now let $C_{i^*}$ be the cluster containing $j^*$.  Define $W = W_{j^*}$, $U = W \cap C_{i^*}$, $V = W \setminus U$.  Thus we have
\[P_k(A)\1_U=\frac{|U|}{s}\1_{C_{i^*}},\; P_k(A)\1_V=\sum_{i \neq i^*}a_i\1_{C_i},\;  0\le a_i, \sum_{i \neq i^*} a_i=\frac{s-|U|}{s}.\]
By~(\ref{jstar}) we have $|U| \geq (1 - 8\epsilon^2)s$, so
\[||P_k(A)\1_{W}||_2^2 = ||P_k(A)\1_U||_2^2 + ||P_k(A)\1_V||_2^2 \geq ||P_k(A)\1_U||_2^2 = \frac{|U|^2}{s} \geq (1 - 8\epsilon^2)^2s.\]
Finally, note that by Lemma~\ref{fundestdifprojnorm} we have $||P_k(\hat A) - P_k(A)||_2 \leq \epsilon$, so the triangle inequality yields the desired result:
\begin{eqnarray*}
||P_k(\hat A)\1_{W}||_2	& \geq 	& ||P_k(A)\1_{W}||_2 - ||(P_k(\hat A) - P_k(A))\1_{W}||_2	\\
& \geq 	& ||P_k(A)\1_{W}||_2 - ||P_k(\hat A) - P_k(A)||_2||\1_W||_2	\\
& \geq 	& (1 - 8\epsilon^2)\sqrt s - \epsilon\sqrt s.
\end{eqnarray*}
\end{proof}

\begin{lemma}\label{almostcluster}
Assume $\hat A$ satisfies~(\ref{l2normbound}) and $j$ satisfies~(\ref{goodcoldef}).  Then $|W_j \cap C_i| \geq (1 - 3\epsilon)s$ for some $i \in [k]$.  
\end{lemma}

\begin{proof}
For $W \subseteq [n]$ define
\[t(W) := \max_{i = 1}^k|C_i \cap W|, \quad \tau := \min_Wt(W),\]
where the minimum is taken over all $W \subseteq [n]$ such that $|W| = s$ and 
\begin{equation}\label{Wcondition}
||P_k(A)\1_W||_2 \geq (1 - 8\epsilon^2 - 2\epsilon)\sqrt s.
\end{equation}
We will argue that
\[t(W_j) \geq \tau \geq (1 - 3\epsilon)s,\]
which proves the lemma.  

The first inequality is easy: for $W = W_j$, the triangle inequality yields
\begin{eqnarray*}
||P_k(A)\1_W||_2	& \geq	& ||P_k(\hat A)\1_W||_2 - ||P_k(A) - P_k(\hat A)||_2||\1_W||_2	\notag \\
& \geq	& (1 - 8\epsilon^2 - \epsilon)\sqrt s - \epsilon\sqrt s		\notag	\\
& =		& (1 - 8\epsilon^2 - 2\epsilon)\sqrt s.\label{P1Wtriangle}
\end{eqnarray*}
So $t(W_j) \geq \tau$.  We just need to show that $\tau \geq (1 - 3\epsilon)s$.

For $W \subseteq [n], i \in [k]$, let $t_i(W) := |C_i \cap W|$.  Construct $W \subseteq [n]$ such that
\begin{itemize}
\item	$|W| = s$.
\item	$W$ satisfies~(\ref{Wcondition}).
\item	$t(W) = \tau$.
\item	$||P_k(A)\1_W||_2$ is as large as possible.
\end{itemize}
We will argue that $W$ must have a special structure: namely, it is split between two clusters.

By relabeling the clusters, we may assume without loss of generality that 
\[\tau = t(W) = t_1(W) \geq \ldots \geq t_k(W).\]
We claim that 
\begin{enumerate}

\item	$t_2(W) < t_1(W)$.  Suppose to the contrary. 
Maximize $\sum_{i=1}^k x_i^2$, such that $x_1=x_2\ge x_3\ge \ldots \ge x_k\ge 0$ and $\sum_{i=1}^kx_i=s$.  It is easy show that  the maximum occurs when $x_1=x_2=\frac{s}{2}$ and $x_3=\ldots=x_k=0$.  Hence $\sum_{i=1}^k t_i(W)^2\le \frac{s^2}{2}$.  By~(\ref{Wcondition}) we have
\[(1 - 8\epsilon^2 - 2\epsilon)^2s \leq ||P_k(A)\1_W||_2^2 = \frac1s\sum_{i = 1}^kt_i(W)^2 \leq \frac s2,\]
which is equivalent to $(1 - 8\epsilon^2 - 2\epsilon) \leq \frac1{\sqrt 2}$.  Choosing $\epsilon$ sufficiently small ($\epsilon \leq .1$ works) we get a contradiction. 

\item	$t_3(W) = \ldots = t_k(W) = 0$.  Assume this is not the case.  Then
\[\tau = t(W) = t_1(W) > t_2(W) \geq t_3(W) \geq 1, \quad \sum_{i =1}^kt_i(W) = s.\]
In particular, $C_3 \cap W$ and $C_2 \setminus W$ are both nonempty.  Now construct $\tilde W$ from $W$ by replacing a vertex from $C_3 \cap W$ with one from $C_2 \setminus W$.  Clearly $|\tilde W| = s$, and $t(\tilde W) = \tau$ since only $t_2$ increases and $t_2(W) < t_1(W) = \tau$.  But
\begin{eqnarray*}
||P_k(A)\1_{\tilde W}||_2^2 - ||P_k(A)\1_W||_2^2	& =	& \frac1s\sum_{i = 1}^kt_i(\tilde W)^2 - \frac1s\sum_{i = 1}^kt_i(W)^2	\\
& =		& \frac1s\left[(t_2(W) + 1)^2 + (t_3(W) - 1)^2 - t_2(W)^2 - t_3(W)^2\right]		\\
& =		& \frac2s(t_2(W) - t_3(W) + 1)	\\
& >		& 0,
\end{eqnarray*}
contradicting the maximality of $||P_k(A)\1_W||_2$.
\end{enumerate}

Thus, $W$ is split between two clusters $C_1$ and $C_2$; i.e., $W = U \cup V$, where $U := W \cap C_1$ and $V := W \cap C_2$.  So by~(\ref{Wcondition}) we have
\begin{equation*}\label{Uquadratic}
(1 - 8\epsilon^2 - 2\epsilon)^2s \leq ||P_k(A)\1_W||_2^2 = ||P_k(A)(\1_U + \1_V)||_2^2 = \frac{|U|^2}{s} + \frac{(s - |U|)^2}{s}.
\end{equation*}
Solving the inequality for $|U|$ yields
\[\tau = \max\{|U|, |V|\} \geq (1 - 3\epsilon)s,\]
provided $\epsilon$ is small enough (again $\epsilon \leq .1$ is sufficient).  This completes the proof.
\end{proof}


\begin{lemma}\label{recoverclusterswhp}
Consider cluster $C_i$ and vertex $j \in [n]$.  If $j \in C_i$, then
\begin{equation}\label{jinCi}
|N_{\hat G}(j) \cap C_i| \geq (p - \epsilon)s
\end{equation}
with probability $\geq 1 - e^{-\epsilon^2s}$, and if if $j \notin C_i$, then
\begin{equation}\label{jnotinCi}
|N_{\hat G}(j) \cap C_i| \leq (q + \epsilon)s
\end{equation}
with probability $\geq 1 - e^{-\epsilon^2s}$.
\end{lemma}

\begin{proof}
Let $j \in C_i$.  Then $\expval[|N(j) \cap C_i|] = p(s - 1)$, so Hoeffding's inequality yields
\[\prb[|N(j) \cap C_i| \leq (p - \epsilon)s] \leq e^{-2(\epsilon s - p)^2/ (s - 1)} \leq e^{-\epsilon^2s}\]
for $n$ (hence $s$) sufficiently large.  On the other hand, if $j \notin C_i$.  Then $\expval[|N(j) \cap C_i|] = qs$, so
\[\prb[|N(j) \cap C_i| \geq (q + \epsilon)s] \leq e^{-2\epsilon^2s} \leq e^{-\epsilon^2s}.\qedhere\]
\end{proof}

\begin{lemma}\label{recovercluster}
Let $W \subseteq [n]$ such that $|W| = s$ and $|W \cap C_i| \geq (1 - 3\epsilon)s$ for some $i \in [k]$.  Then
\begin{enumerate}[a)]
\item	If $j \in C_i$ and $j$ satisfies~(\ref{jinCi}), then $|N_{\hat G}(j) \cap W| \geq (p - 4\epsilon)s$.
\item	If $j \in [n] \setminus C_i$ and $j$ satisfies~(\ref{jnotinCi}), then $|N_{\hat G}(j) \cap W| \leq (q + 4\epsilon)s$.\label{recoverclusterb}
\end{enumerate}
\end{lemma}

\begin{proof}
Assume $j \in C_i$ and $j$ satisfies~(\ref{jinCi}).  As $|C_i| = s$, we have $|C_i \setminus W| \leq 3\epsilon$.  Therefore,
\begin{eqnarray*}
|N(j) \cap W|	& \geq	& |N(j) \cap W \cap C_i|	\\
& =		& |N(j) \cap C_i| - |(N(j) \cap C_i) \setminus W|	\\
& \geq	& |N(j) \cap C_i| - |C_i \setminus W|	\\
& \geq	& (p - \epsilon)s - 3\epsilon s	\\
& =		& (p - 4\epsilon)s.
\end{eqnarray*}
Part~\ref{recoverclusterb}) follows by a similar argument.
\end{proof}

This lemma gives us a way to differentiate between vertices $j \in C_i$ and vertices $j \notin C_i$ as shown in Figure~\ref{fig:recovercluster}, provided 
\begin{equation}\label{epsiloncondition}
p - 4\epsilon \geq q + 4\epsilon.
\end{equation}

\begin{figure}
\centering
\begin{tikzpicture}[scale=.08]
\draw (0, 0) ellipse(10 and 20);
\draw (-8.66, -10) -- (8.66, -10);
\draw (0, 0) node[anchor=south]{$W \cap C_i$};
\draw (0, 20) node[anchor=south]{$W$};
\foreach \x in {-30, 30}
\fill (\x, 0) circle(1);
\draw (-30, 0) node[anchor=east]{\begin{tabular}{r}$j \in C_i$ \\ $\geq (p - 4\epsilon)s$\end{tabular}};
\draw (30, 0) node[anchor=west]{\begin{tabular}{l}$j \notin C_i$ \\ $\leq (q + 4\epsilon)s$\end{tabular}};
\foreach \y in {-18, -14, ..., 18}
\draw (-30, 0) -- (-1, \y);
\foreach \y in {-18, -9, 0, 9, 18}
\draw[densely dotted] (30, 0) -- (1, \y);
\end{tikzpicture}
\caption{If $W$ has large overlap with $C_i$, then a.s.\ vertices in $C_i$ will have many neighbors in $W$, while vertices not in $C_i$ will have relatively few neighbors in $W$.}\label{fig:recovercluster}
\end{figure}
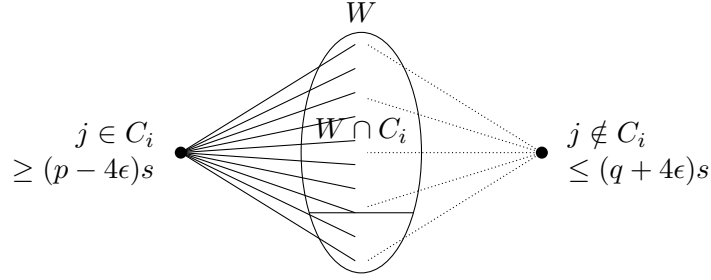

\subsection{Main proof}\label{mainpf}

To prove Theorem~\ref{mainthm}, we will define certain (exponentially many) events on the probability space $\mathcal G(n, \mathcal C, p, q)$ and show that
\begin{enumerate}
\item	As long as they all occur, Algorithm~\ref{alg} \emph{definitely} succeeds.
\item	They all occur simultaneously a.s.
\end{enumerate}
Therefore, Algorithm~\ref{alg} succeeds a.s.

Before we define the events let us introduce some notation:
\begin{itemize}
\item	For $J \subseteq [k]$,  define $\hat G_J$ to be the subgraph of $\hat G$ induced by clusters $C_i, i \in J$, i.e.\ $\hat G_J := \hat G\left[\bigcup_{i \in J}C_i\right]$.  Then for any fixed $J$ we have 
\begin{equation}\label{induceddistr}
\hat G_J \sim \mathcal G(|J|s, \{C_i : i \in J\}, p, q).
\end{equation}
\item	For an $n \times n$ matrix $B$ define $B_J$ to be principal submatrix of $B$ with row and column indices in the clusters $C_i, i \in J$, i.e. $B_J :=  B\left[\bigcup_{i \in J}C_i\right]$.  
\end{itemize}
We will refer to these subgraphs and submatrices as \emph{cluster subgraphs} and \emph{cluster submatrices}.

Now we define two types of events in $\mathcal G(n, \mathcal C, p, q)$:
\begin{itemize}
\item	\emph{Spectral events}: for $J \subseteq [k]$, let $E_J$ be the event that $||\hat A_J - A_J||_2 \leq 8\sqrt{|J|s}$.
\item	\emph{Degree events}: for $1 \leq i \leq k, 1 \leq j \leq n$, let $D_{i, j}$ be the event that $|N_{\hat G}(j) \cap C_i| \geq (p - \epsilon)s$ if $j \in C_i$, or the event that $|N_{\hat G}(j) \cap C_i| \leq (q + \epsilon)s$ if $j \notin C_i$.
\end{itemize}
Thus, we have defined a total of $2^k + nk$ events.  Essentially, these are the events that every $\hat G_J$ satisfies~(\ref{l2normbound}) and that~(\ref{jinCi}) and~(\ref{jnotinCi}) are satisfied for all $i \in [k], j \in [n]$.  Note that the events are well-defined, as their definitions depend only on the underlying probability space $\mathcal G(n, \mathcal C, p, q)$ and not on the random graph $\hat G$ sampled from the space.

Now we are finally ready to prove the theorem:

\begin{proof}[Proof of Theorem~\ref{mainthm}]

Assume $E_J$ and $D_{i, j}$ hold for all $J \subseteq [k], i \in [k], j \in [n]$.  We will prove by induction that Algorithm~\ref{alg} succeeds in every iteration.

For the base case, take the original graph $\hat G = \hat G_{[k]}$ considered in the fist iteration.  Since $E_{[k]}$ is assumed to hold, (\ref{l2normbound}) is satisfied.  Thus, by Lemma~\ref{existsgoodcol}, the column $j = j^*$ identified in step \ref{step:jstar} satisfies~(\ref{goodcoldef}).  Then by Lemma~\ref{almostcluster} we have $|W_{j^*} \cap C_i| \geq (1 - 3\epsilon)s$ for some $i \in [k]$.  Finally, since $D_{i, j}$ is assumed to hold for all $j \in [n]$, step~\ref{step:id} correctly identifies $C = C_i$ by Lemma~\ref{recovercluster}.

Now assume Algorithm~\ref{alg} succeeds in the first $t$ iterations, i.e., it correctly identifies a cluster and removes it in each of these iterations.  Then the graph considered in the $(t + 1)$st iteration is a cluster subgraph $\hat G_J$ for some $J \subseteq [k]$, $|J| = k - t$.  Note that $\hat G_J$ has $|J|s = (k - t)s$ vertices.  Now we apply Lemmas~\ref{Ahat-eigs}-\ref{almostcluster} with $\hat A_J$ instead of $\hat A$, $A_J$ instead of $A$, $k - t$ instead of $k$, and $(k - t)s$ instead of $n$.

Since $E_J$ is assumed to hold, by Lemma~\ref{existsgoodcol} the column $j = j^*$ identified in step \ref{step:jstar} of Algorithm~\ref{alg} satisfies $||P_{k - t}(\hat A_J)\1_{W_j}||_2 \geq (1 - 8\epsilon^2 - \epsilon)\sqrt s$.  Note that $\hat H$ and $W_j$ (Sections~\ref{mainpfnotation}-\ref{techlemmas}) are constructed from $\hat G_J$, not the original graph $\hat G$.  Now by Lemma~\ref{almostcluster} we have $|W_{j^*} \cap C_i| \geq (1 - 3\epsilon)s$ for some $i \in J$.  Finally, since $D_{i, j}$ is assumed to hold for all $j \in [n]$, step~\ref{step:id} once again correctly identifies $C = C_i$ by Lemma~\ref{recovercluster}.

We have thus proved that Algorithm~\ref{alg} succeeds as long as $E_J$ and $D_{i, j}$ hold for all $J \subseteq [k], i \in [k], j \in [n]$.  Now, for any fixed nonempty $J \subseteq [k]$ we have $\hat G_J \sim \mathcal G(|J|s, \{C_i : i \in J\}, p, q)$, so by Lemma~\ref{normAhat-A}
\[\prb[{E}_J] \geq 1-e^{-|J|s} \geq 1-e^{-s}.\]
By Lemma~\ref{recoverclusterswhp}, for any $i, j$
\[\prb[D_{i, j}] \geq 1 - e^{-\epsilon^2s.}\]
Taking a union bound over all $J, i, j$, the probability that all $E_J$ and $D_{i, j}$ hold is $\geq 1 - 2^ke^{-s} - nke^{-\epsilon^2s}$.  Therefore, as $\epsilon$ is constant and $k \leq \sqrt n \leq s$, Algorithm~\ref{alg} succeeds with probability $\geq 1 - \left(\frac2e\right)^{-\sqrt n} - n^{3 / 2}e^{-\sqrt n}$.

Note that we require~(\ref{epsiloncondition}) in order for step~\ref{step:id} of Algorithm~\ref{alg} to correctly recover a cluster according to Lemma~\ref{recovercluster}.  In addition, the proof of Lemma~\ref{almostcluster} requires $\epsilon \leq .1$.  By~(\ref{sreq2}), we can satisfy both of these conditions by setting $c := \max\left\{\frac{88}{p - q}, \frac{72}{(p-q)^2}\right\}$.
\end{proof}

\section*{Acknowledgements}

We would like to thank the anonymous reviewers of this and previous versions of our paper for pointing us to relevant past work and 
for calling our attention to the fact that the iterations of Algorithm~\ref{alg} cannot be handled independently.
The latter issue fixed herein 
by ``preprocessing the randomness'' as discussed in Sections~\ref{sec:algsection} and~\ref{sec:clusteridproof}.  

This research was partly supported by ARO grant 66497-NS and NSF grant IIS-1526379.
\bibliography{paper}
\bibliographystyle{plain}

  \appendix
  \counterwithin{thm}{section}

\section{Comparison with previous results}\label{prevwork}

The following table compares our work with previous algorithms for recovering planted partitions.  Note that some of the algorithms apply to more general planted clustering settings, but here we list their performance \emph{only in the setting considered in this paper}, i.e.~for planted partition with constant edge probabilities and uniform cluster sizes.  In particular, note that the the well-known results~\cite{McSherry01, vu2014simple} achieve the $\sqrt n$ bound in certain settings, but only when some of the clusters are size $\omega(\sqrt n)$.

\begin{center}
\begin{tabular}{|lrl|l|l|}
\hline
Paper	& 	&	& Minimum cluster size 	& Algorithm type			\\
& & &  for planted partition & \\
\hline
McSherry 	& 2001	&  \cite[Theorem~4]{McSherry01}		& $\Omega(n^{2 /3})$	& Spectral			\\
Giesen \& Mitsche	& 2005	& \cite{Giesen05}	& $\Omega(\sqrt n)$	& Spectral	\\
Oymak \& Hassibi	& 2011	& \cite{oymak2011}		& $\Omega(\sqrt n)$	& Convex programming	\\
Ames			& 2014	& \cite{ames2014}	& $\Omega(\sqrt n)$		& Semidefinite programming	\\
Chen et al.	& 2014	& \cite{chen2014}	& $\Omega(\sqrt n)$	& Convex programming	\\
Vu			& 2014	& \cite[Theorem~2]{vu2014simple}	& $\omega(\sqrt{n\log n})$	& Spectral	\\
			&		&					&									&			\\
Our result	&		&					& $\Omega(\sqrt n)$	& Spectral	\\
\hline
\end{tabular}
\end{center}

Thus, we see that while many have succeeded in recovering clusters of size $\Omega(\sqrt n)$, prior to this paper, only Giesen and Mitsche \cite{Giesen05} had done so using a purely spectral approach (to the best of our knowledge).  While their proof techniques have much in common with our own, our algorithm is arguably much simpler.  Furthermore, their algorithm takes $2^{\Theta(k)} \cdot \poly(n)$ time and is thus inefficient when $k =\omega(\log n)$.

%
%
%
%
%
%

  \end{document}